\documentclass[english]{article}
\usepackage[T1]{fontenc}
\usepackage[utf8]{inputenc}
\usepackage{amsmath,amsthm} 
\usepackage{graphicx}
\usepackage{amssymb}
\usepackage{enumitem}
\usepackage{natbib} 
\setcitestyle{round}
\usepackage{bussproofs}
\usepackage[letterpaper]{geometry} 
\usepackage{thmtools,thm-restate}
\usepackage{tikz}
\usepackage{pgf} 
\usetikzlibrary{patterns,automata,arrows,shapes,snakes,topaths,trees,backgrounds,positioning,through,calc}
\usepackage{setspace}
\usepackage{multicol}
\usepackage{stmaryrd}  
\usepackage{colortbl}

\theoremstyle{definition}
\newtheorem{theorem}{Theorem}

\newtheorem{definition}[theorem]{Definition}
\newtheorem{lemma}[theorem]{Lemma}
  
\newtheorem{remark}[theorem]{Remark}

\usepackage{hyperref}
\hypersetup{colorlinks=true, citecolor=blue, linkcolor=blue}  

\begin{document} 

\title{The incompatibility of the Condorcet winner and loser criteria with positive involvement and resolvability}

\author{Wesley H. Holliday \\ University of California, Berkeley \\ {\small wesholliday@berkeley.edu}}

\date{{\small Published in \textit{Economics Letters}, Vol.~262, 2026, 112868.}}

\maketitle

\begin{abstract} We prove that there is no preferential voting method satisfying the Condorcet winner and loser criteria, positive involvement (if a candidate $x$ wins in an initial preference profile, then adding a voter who ranks $x$ uniquely first cannot cause $x$ to lose), and $n$-voter resolvability (if $x$ initially ties for winning, then $x$ can be made the unique winner by adding some set of up to $n$ voters). This impossibility theorem holds for any positive integer $n$. It also holds if either the Condorcet loser criterion is replaced by independence of clones or positive involvement is replaced by negative involvement.\end{abstract}

\section{Introduction}\label{Intro}

In the theory of preferential voting, where voters submit rankings of the candidates, there is a fundamental tradeoff when selecting a voting method between ensuring desired \textit{variable-voter} properties, ensuring desired \textit{variable-candidate} properties, and ensuring that ties are not too common. In this note, we prove an impossibility theorem illustrating this tradeoff. The relevant variable-voter axiom is that of \textit{positive involvement} (\citealt{Saari1995}, \citealt{Perez2001}, \citealt{HP2021b}): if a candidate $x$ wins in an initial preference profile, then adding a voter who ranks $x$ uniquely first cannot cause $x$ to lose. Alternatively, we can use \textit{negative involvement} (\citealt{Fishburn1983}, \citealt{Perez2001}): if a candidate $x$ loses in an initial preference profile, then adding a voter who ranks $x$ uniquely last cannot cause $x$ to win.  According to P\'{e}rez \citeyearpar[p.~605]{Perez2001}, these axioms ``may be seen as the minimum to require concerning the coherence in the winning set when new voters~are~added.''\footnote{However, a weaker requirement is \textit{positive-negative involvement}: adding a voter who ranks an initial winner $x$ uniquely first and an initial loser $y$ uniquely last cannot cause $x$ to lose and $y$ to win. Another weaker requirement is \textit{bullet-vote positive involvement}: adding a voter who ranks an initial winner $x$ uniquely first---and leaves all other candidates unranked or ranks them as tied after $x$---cannot cause $x$ to lose.} On the candidate side, although the famous Condorcet winner and loser criteria (\citealt{Condorcet1785}) are not explicitly variable-candidate axioms, they can be motivated by considering candidates dropping out or joining the race: one might think that if $x$ would win (resp.~lose) in every two-candidate sub-election, then $x$ should still win (resp.~lose) in the full election---or else the election outcome might be seen as too contingent on which other losing candidates happened to stay in the race. Instead of the Condorcet loser criterion, we can use the explicitly variable-candidate axiom of \textit{independence of clones} (\citealt{Tideman1987}). Finally, for a formalization of the idea that a voting method should not make ties too common, we use the \textit{resolvability} axiom (\citealt{Tideman1986,Tideman1987}) that if $x$ is initially tied for winning, then $x$ can be made the \textit{unique} winner by adding a single new voter. As Tideman \citeyearpar[p.~25]{Tideman1986} puts it, ``Resolvability may be thought of as a condition that ties happen only `by coincidence', every tie being within one vote of being broken.'' In fact, we will ultimately consider, for any positive integer $n$, the axiom of \textit{$n$-voter resolvability}, whereby any tied winner can be made the unique winner by adding some set of up to $n$ new~voters.

\section{Setup}

Given infinite sets $\mathcal{V}$ and $\mathcal{X}$ of possible voters and possible candidates, respectively, a \textit{profile} $\mathbf{P}$ is a function assigning to each voter $i$ in some nonempty finite set $V(\mathbf{P})\subseteq\mathcal{V}$ a strict weak order $\mathbf{P}(i)$ on some nonempty finite set $X(\mathbf{P})\subseteq\mathcal{X}$ of candidates. We call $\mathbf{P}$ \textit{linear} if for each $i\in V(\mathbf{P})$,  $\mathbf{P}(i)$ is a linear order. All results to follow continue to hold if all profiles are linear. If $\mathbf{P}$ and $\mathbf{Q}$ are profiles with $V(\mathbf{P})\cap V(\mathbf{Q})=\varnothing$ and $X(\mathbf{P})=X(\mathbf{Q})$, then $\mathbf{P}+\mathbf{Q}$ is the profile with $V(\mathbf{P}+\mathbf{Q})=V(\mathbf{P})\cup V(\mathbf{Q})$ such that for $i\in V(\mathbf{P})$, $(\mathbf{P}+\mathbf{Q})(i)=\mathbf{P}(i)$, and for $i\in V(\mathbf{Q})$,  $(\mathbf{P}+\mathbf{Q})(i)=\mathbf{Q}(i)$.

Given $x,y\in X(\mathbf{P})$, the \textit{margin of $x$ vs.~$y$} in $\mathbf{P}$ is the number of voters who rank $x$ above $y$ minus the number who rank $y$ above $x$: $\text{Margin}_\mathbf{P}(x,y)=|\{i\in V(\mathbf{P})\mid (x,y)\in \mathbf{P}(i)\}| - |\{i\in V(\mathbf{P})\mid (y,x)\in\mathbf{P}(i)\}|$. The \textit{margin graph} $\mathcal{M}(\mathbf{P})$ of $\mathbf{P}$ is the weighted directed graph whose set of vertices is $X(\mathbf{P})$ with an edge from $x$ to $y$ when $\text{Margin}_\mathbf{P}(x,y)>0$, weighted by $\text{Margin}_\mathbf{P}(x,y)$.

A voting method is a function $F$ assigning to each profile $\mathbf{P}$ a nonempty subset $F(\mathbf{P})\subseteq X(\mathbf{P})$. This builds an assumption of ``universal domain'' into the definition of a voting method.

\begin{definition}\label{AxDef} Let $F$ be a voting method.
\begin{enumerate}
\item\label{CW} $F$ satisfies the \textit{Condorcet winner criterion} if for every profile $\mathbf{P}$ and $x\in X(\mathbf{P})$, if $\text{Margin}_\mathbf{P}(x,y)>0$ for all $y\in X(\mathbf{P})\setminus\{x\}$ (so $x$ is a \textit{Condorcet winner}), then $F(\mathbf{P})=\{x\}$.
\item\label{CL} $F$ satisfies the \textit{Condorcet loser criterion} if for every profile $\mathbf{P}$ and $x\in X(\mathbf{P})$, if $|X(\mathbf{P})|>1$  and $\text{Margin}_\mathbf{P}(y,x)>0$ for all $y\in X(\mathbf{P})\setminus\{x\}$ (so $x$ is a \textit{Condorcet loser}), then $x\not\in F(\mathbf{P})$.
\item\label{PosInv} $F$ satisfies \textit{positive involvement} if for every profile $\mathbf{P}$, $x\in X(\mathbf{P})$, and profile $\mathbf{P}^x$ with $V(\mathbf{P})\cap V(\mathbf{P}^x)=\varnothing$ and $X(\mathbf{P})=X(\mathbf{P}^x)$, if $x\in F(\mathbf{P})$ and $\mathbf{P}^x$ consists of a single voter who ranks $x$ uniquely first, then $x\in F(\mathbf{P}+\mathbf{P}^x)$.
\item\label{resolve} $F$ satisfies \textit{resolvability} if for every profile $\mathbf{P}$, if $|F(\mathbf{P})|>1$, then for each $x\in F(\mathbf{P})$, there is some single-voter profile $\mathbf{S}$ with $V(\mathbf{P})\cap V(\mathbf{S})=\varnothing$ and $X(\mathbf{P})=X(\mathbf{S})$ such that $F(\mathbf{P}+\mathbf{S})=\{x\}$.\footnote{In \citealt{Holliday2024}, we used a weaker version of this axiom in the impossibility theorem, since there the weaker version sufficed: if $|F(\mathbf{P})|>1$, then there is some single-voter profile $\mathbf{S}$ such that $|F(\mathbf{P}+\mathbf{S})|=1$. But Tideman's \citeyearpar{Tideman1986} original axiom requires that \textit{for each $x\in F(\mathbf{P})$}, there is some single-voter profile $\mathbf{S}$ such that $F(\mathbf{P}+\mathbf{S})=\{x\}$. This version facilitates the proof in this paper (\citealt{Tideman1986} does not include the assumption that $|F(\mathbf{P})|>1$, but the language of \citealt{Tideman1987} suggests it). In \citealt{Holliday2024}, we also discussed \textit{asymptotic resolvability}, but without the assumption of ordinal margin invariance made in that paper, asymptotic resolvability is more difficult to work with, so we do not discuss it here.}
\end{enumerate}
\end{definition}

Each of the axioms in Definition~\ref{AxDef} is independent of the others. The Borda count (\citealt{Borda1784}) satisfies all the axioms except the Condorcet winner criterion. Minimax (\citealt{Simpson1969}, \citealt{Kramer1977}) satisfies all of them except the Condorcet loser criterion.\footnote{Minimax does not satisfy resolvability if we restrict its domain to linear profiles (so the added voter must submit a linear order), but the related \textit{leximax} method (see \citealt{Brandt2025})---wherein the winning candidates are those whose vectors of margins, sorted in increasing order, are lexicographically largest---does satisfy resolvability when its domain is restricted (or not restricted) to linear profiles, as well as the Condorcet winner criterion and positive involvement (for all profiles).} Ranked Pairs (\citealt{Tideman1987}) satisfies all of them except positive involvement. Split Cycle (\citealt{HP2023}) satisfies all of them except resolvability. 

Below we prove that there is no voting method satisfying all the axioms in Definition~\ref{AxDef}. The proof involves reasoning about profiles with five candidates. It extends the strategy of \citealt{Holliday2024}, which showed that for four candidates, the axioms are jointly inconsistent with an additional axiom of \textit{ordinal margin invariance}. We leave open whether the axioms are jointly satisfiable without ordinal margin invariance if we only ask for a voting method defined on four-candidate profiles. If we restrict the domain to two- or three-candidate profiles, then the axioms are satisfied by the Minimax voting method.

The impossibility proof in the next section makes use of the following concept (cf.~\citealt{Kasper2019}, \citealt{Holliday2024}).

\begin{definition}\label{DefenseDef} Given a profile $\mathbf{P}$, the \textit{defensible set $D(\mathbf{P})$ of $\mathbf{P}$} is the set of candidates $x\in X(\mathbf{P})$ such that for all $y\in X(\mathbf{P})$, there is a $z\in X(\mathbf{P})$ such that $\text{Margin}_\mathbf{P}(z,y)\geq \text{Margin}_\mathbf{P}(y,x)$.
\end{definition}

The following lemma is essentially Lemma 3 of \citealt{Perez1995} (cf.~Lemma 1 of \citealt{Perez2001}), which is in turn an adaptation of a similar result of \citealt{Moulin1988}, but we include a proof for convenience.

\begin{lemma}[\citealt{Moulin1988}, \citealt{Perez1995}]\label{RefineLem} Let $F$ be a voting method satisfying the Condorcet winner criterion and positive involvement. Let $\mathbf{P}$  be a profile in which for all $x,x',y,y'\in X(\mathbf{P})$ with $x\neq x'$ and $y\neq y'$, if $\text{Margin}_\mathbf{P}(x,x')>\text{Margin}_\mathbf{P}(y,y')$, then $\text{Margin}_\mathbf{P}(x,x')>\text{Margin}_\mathbf{P}(y,y')+1$. Then $F(\mathbf{P})\subseteq D(\mathbf{P})$.
\end{lemma}

\begin{proof} Let $F$ be such a method and $x\in F(\mathbf{P})$. For contradiction, suppose $x\not\in D(\mathbf{P})$. Then there is a $y\in X(\mathbf{P})$ such that for every $z\in X(\mathbf{P})$, $\text{Margin}_\mathbf{P}(z,y)< \text{Margin}_\mathbf{P}(y,x)$, which implies $y\neq x$. Let $k= \max\{\text{Margin}_\mathbf{P}(z,y) \mid z\in X(\mathbf{P}),z\neq y\}$, so $k<\text{Margin}_\mathbf{P}(y,x)$. If $k< 0$, then $y$ is the Condorcet winner in $\mathbf{P}$, so $F(\mathbf{P})=\{y\}$ by the Condorcet winner criterion, which contradicts $x\in F(\mathbf{P})$. Thus, $k\geq 0$. By our assumption about $\mathbf{P}$ in the lemma, $k+1<\text{Margin}_\mathbf{P}(y,x)$. Let $\mathbf{P}'$ be obtained from $\mathbf{P}$ by adding $k+1$ voters who rank $x$ uniquely first and $y$ uniquely second, followed by any linear order of $X(\mathbf{P})\setminus \{x,y\}$. Hence by repeated application of positive involvement, $x\in F(\mathbf{P}')$. But $y$ is the Condorcet winner in $\mathbf{P}'$, so $F(\mathbf{P}')=\{y\}$ by the Condorcet winner criterion, which contradicts $x\in F(\mathbf{P}')$. Thus, $x\in D(\mathbf{P})$.\end{proof}

\section{Result}

\begin{theorem}\label{MainThm} There is no voting method satisfying the Condorcet winner and loser criteria, positive involvement, and resolvability.
\end{theorem}

\begin{proof} Suppose for contradiction that there is such a voting method $F$. Let $\mathbf{P}_1$ be a profile as in Figure \ref{Prof1}, whose margin graph is $\mathcal{M}_1$ in Figure \ref{KeyFig}.\footnote{\label{P1note}Rather than using this specific $\mathbf{P}_1$, one may reason more abstractly as follows. By Debord's Theorem (\citealt{Debord1987}), there is a profile $\mathbf{Q}$ whose margin graph is $\mathcal{M}_1$. Now let $\mathbf{Q}^+$ be obtained from $\mathbf{Q}$ by adding, for each ranking of the five candidates, $7$ new voters with that ranking; then the margin graph of $\mathbf{Q}^+$ is still $\mathcal{M}_1$, and $\mathbf{Q}^+$ contains all the rankings needed for the rest of the argument, so one may use $\mathbf{Q}^+$ in place of $\mathbf{P}_1$. After the initial construction of this proof using another version of $\mathcal{M}_1$ and different numbers of rankings to add/subtract to obtain an appropriate sequence $\mathcal{M}_1,\dots,\mathcal{M}_5$ of margin graphs for the proof, an integer linear program was used to find a small profile $\mathbf{P}_1$ and numbers of rankings to add/subtract to keep the total number of voters down while still satisfying all the constraints needed for each step of the proof. Even fewer voters suffice (e.g., $230$ instead of $261$ in $\mathbf{P}_4$, the largest $\mathbf{P}_i$) if, at the expense of elegance, we allow more distinct ranking types in the profiles.}  Let $\mathbf{P}_2,\dots,\mathbf{P}_5$ be obtained from $\mathbf{P}_1$ by adding and subtracting the profiles shown in the equations in Figure \ref{KeyFig}: $\mathbf{P}_2$ is obtained from $\mathbf{P}_1$ by adding 26 new voters with the ranking $adbec$; $\mathbf{P}_3$ is obtained from $\mathbf{P}_2$ by removing $7$ voters with the ranking $daceb$, etc. Note there are indeed such $7$ voters in $\mathbf{P}_1$, and for the purposes of $\mathbf{P}_5$, there are $7$ voters in $\mathbf{P}_1$ with the ranking $dbace$. Then $\mathcal{M}_2,\dots,\mathcal{M}_5$ are the margin graphs of $\mathbf{P}_2,\dots,\mathbf{P}_5$, respectively. 

\begin{figure}[h]
\begin{center}
  \begin{tabular}{cccccc}
  69 & 64 & 46 & 20 & 18 & 2\\
  \hline
  $d$ & $e$ & $b$ & $c$ & $d$ & $e$ \\
  $a$ & $b$ & $c$ & $d$ & $b$ & $d$ \\
  $c$ & $a$ & $a$ & $e$ & $a$ & $c$ \\
  $e$ & $c$ & $e$ & $b$ & $c$ & $b$ \\
  $b$ & $d$ & $d$ & $a$ & $e$ & $a$ \\
  \end{tabular}
  \end{center}
\caption{the profile $\mathbf{P}_1$ used in the proof of Theorem \ref{MainThm}. The number at the top of each column is the number of voters submitting the ranking in that column.}\label{Prof1}
\end{figure}

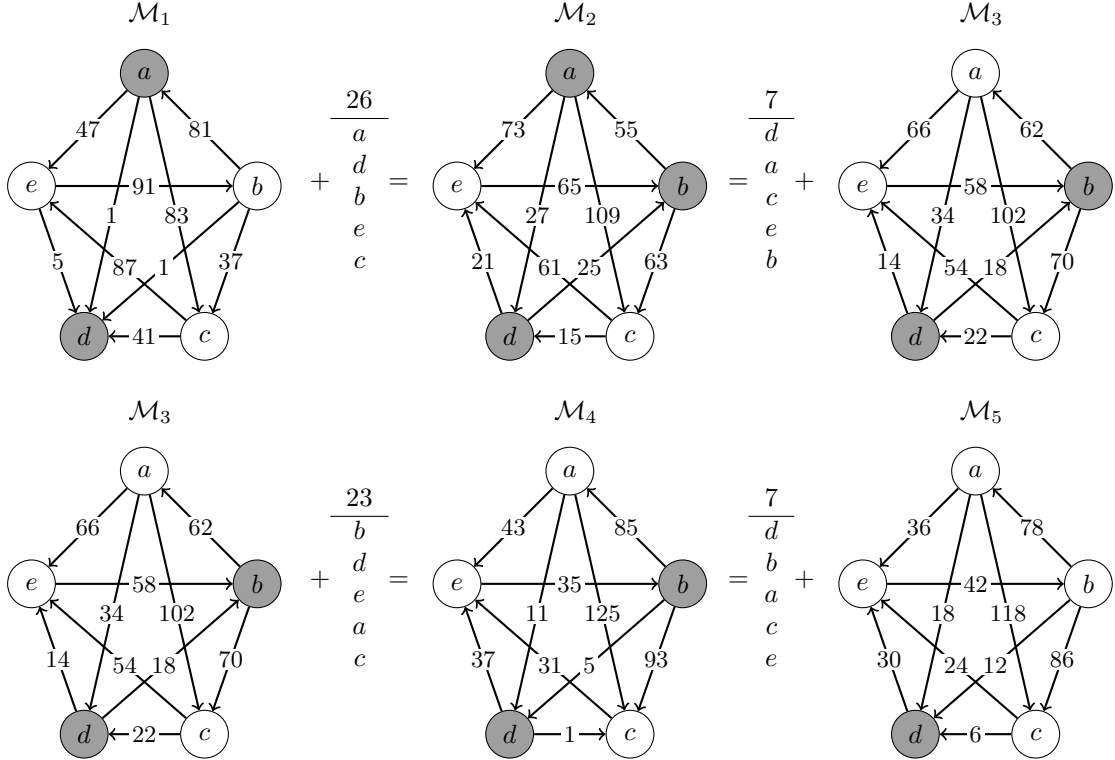
\begin{figure}[h!]
 \begin{center}

  \begin{minipage}{1.5in}
  \begin{center}
  $\mathcal{M}_1$
  \end{center}
  \begin{tikzpicture}

  \node[circle,draw,fill=gray!75, minimum width=0.25in] at (1.5,3) (a) {$a$};
  \node[circle,draw,minimum width=0.25in] at (3,1.5) (b) {$b$};
  \node[circle,draw,minimum width=0.25in] at (2.3,-0.5) (c) {$c$};
  \node[circle,draw,fill=gray!75, minimum width=0.25in] at (0.7,-0.5) (d) {$d$};
  \node[circle,draw,minimum width=0.25in] at (0,1.5) (e) {$e$};

  \path[->,draw,thick] (a) to node[fill=white,inner sep=1pt,pos=0.55] {\small $83$} (c);
  \path[->,draw,thick] (a) to node[fill=white,inner sep=1pt,pos=0.55] {\small $1$} (d);
  \path[->,draw,thick] (a) to node[fill=white,inner sep=1pt] {\small $47$} (e);
  \path[->,draw,thick] (b) to node[fill=white,inner sep=1pt] {\small $81$} (a);
  \path[->,draw,thick] (b) to node[fill=white,inner sep=1pt] {\small $37$} (c);
  \path[->,draw,thick] (b) to node[fill=white,inner sep=1pt,pos=0.55] {\small $1$} (d);
  \path[->,draw,thick] (c) to node[fill=white,inner sep=1pt] {\small $41$} (d);
  \path[->,draw,thick] (c) to node[fill=white,inner sep=1pt,pos=0.45] {\small $87$} (e);
  \path[->,draw,thick] (e) to node[fill=white,inner sep=1pt] {\small $91$} (b);
  \path[->,draw,thick] (e) to node[fill=white,inner sep=1pt] {\small $5$} (d);

  \end{tikzpicture}

  \end{minipage}
$\;+$\begin{tabular}{c}
  26 \\
  \hline
  $a$ \\
  $d$ \\
  $b$ \\
  $e$ \\
  $c$
\end{tabular}$=\;\;$  \begin{minipage}{1.5in}
  \begin{center}
  $\mathcal{M}_2$
  \end{center}
  \begin{tikzpicture}

  \node[circle,draw,fill=gray!75, minimum width=0.25in] at (1.5,3) (a) {$a$};
  \node[circle,draw,fill=gray!75, minimum width=0.25in] at (3,1.5) (b) {$b$};
  \node[circle,draw,minimum width=0.25in] at (2.3,-0.5) (c) {$c$};
  \node[circle,draw,fill=gray!75, minimum width=0.25in] at (0.7,-0.5) (d) {$d$};
  \node[circle,draw,minimum width=0.25in] at (0,1.5) (e) {$e$};

  \path[->,draw,thick] (a) to node[fill=white,inner sep=1pt,pos=0.55] {\small $109$} (c);
  \path[->,draw,thick] (a) to node[fill=white,inner sep=1pt,pos=0.55] {\small $27$} (d);
  \path[->,draw,thick] (a) to node[fill=white,inner sep=1pt] {\small $73$} (e);
  \path[->,draw,thick] (b) to node[fill=white,inner sep=1pt] {\small $55$} (a);
  \path[->,draw,thick] (b) to node[fill=white,inner sep=1pt] {\small $63$} (c);
  \path[->,draw,thick] (c) to node[fill=white,inner sep=1pt] {\small $15$} (d);
  \path[->,draw,thick] (c) to node[fill=white,inner sep=1pt,pos=0.45] {\small $61$} (e);
  \path[->,draw,thick] (d) to node[fill=white,inner sep=1pt,pos=0.45] {\small $25$} (b);
  \path[->,draw,thick] (d) to node[fill=white,inner sep=1pt] {\small $21$} (e);
  \path[->,draw,thick] (e) to node[fill=white,inner sep=1pt] {\small $65$} (b);

  \end{tikzpicture}

  \end{minipage}
$=$\begin{tabular}{c}
  7 \\
  \hline
  $d$ \\
  $a$ \\
  $c$ \\
  $e$ \\
  $b$
\end{tabular}$+\;\;$  \begin{minipage}{1.5in}
  \begin{center}
  $\mathcal{M}_3$
  \end{center}
  \begin{tikzpicture}

  \node[circle,draw,minimum width=0.25in] at (1.5,3) (a) {$a$};
  \node[circle,draw,fill=gray!75, minimum width=0.25in] at (3,1.5) (b) {$b$};
  \node[circle,draw,minimum width=0.25in] at (2.3,-0.5) (c) {$c$};
  \node[circle,draw,fill=gray!75, minimum width=0.25in] at (0.7,-0.5) (d) {$d$};
  \node[circle,draw,minimum width=0.25in] at (0,1.5) (e) {$e$};

  \path[->,draw,thick] (a) to node[fill=white,inner sep=1pt,pos=0.55] {\small $102$} (c);
  \path[->,draw,thick] (a) to node[fill=white,inner sep=1pt,pos=0.55] {\small $34$} (d);
  \path[->,draw,thick] (a) to node[fill=white,inner sep=1pt] {\small $66$} (e);
  \path[->,draw,thick] (b) to node[fill=white,inner sep=1pt] {\small $62$} (a);
  \path[->,draw,thick] (b) to node[fill=white,inner sep=1pt] {\small $70$} (c);
  \path[->,draw,thick] (c) to node[fill=white,inner sep=1pt] {\small $22$} (d);
  \path[->,draw,thick] (c) to node[fill=white,inner sep=1pt,pos=0.45] {\small $54$} (e);
  \path[->,draw,thick] (d) to node[fill=white,inner sep=1pt,pos=0.45] {\small $18$} (b);
  \path[->,draw,thick] (d) to node[fill=white,inner sep=1pt] {\small $14$} (e);
  \path[->,draw,thick] (e) to node[fill=white,inner sep=1pt] {\small $58$} (b);

  \end{tikzpicture}

  \end{minipage}

  \vspace{0.5cm}

  \begin{minipage}{1.5in}
  \begin{center}
  $\mathcal{M}_3$
  \end{center}
  \begin{tikzpicture}

  \node[circle,draw,minimum width=0.25in] at (1.5,3) (a) {$a$};
  \node[circle,draw,fill=gray!75, minimum width=0.25in] at (3,1.5) (b) {$b$};
  \node[circle,draw,minimum width=0.25in] at (2.3,-0.5) (c) {$c$};
  \node[circle,draw,fill=gray!75, minimum width=0.25in] at (0.7,-0.5) (d) {$d$};
  \node[circle,draw,minimum width=0.25in] at (0,1.5) (e) {$e$};

  \path[->,draw,thick] (a) to node[fill=white,inner sep=1pt,pos=0.55] {\small $102$} (c);
  \path[->,draw,thick] (a) to node[fill=white,inner sep=1pt,pos=0.55] {\small $34$} (d);
  \path[->,draw,thick] (a) to node[fill=white,inner sep=1pt] {\small $66$} (e);
  \path[->,draw,thick] (b) to node[fill=white,inner sep=1pt] {\small $62$} (a);
  \path[->,draw,thick] (b) to node[fill=white,inner sep=1pt] {\small $70$} (c);
  \path[->,draw,thick] (c) to node[fill=white,inner sep=1pt] {\small $22$} (d);
  \path[->,draw,thick] (c) to node[fill=white,inner sep=1pt,pos=0.45] {\small $54$} (e);
  \path[->,draw,thick] (d) to node[fill=white,inner sep=1pt,pos=0.45] {\small $18$} (b);
  \path[->,draw,thick] (d) to node[fill=white,inner sep=1pt] {\small $14$} (e);
  \path[->,draw,thick] (e) to node[fill=white,inner sep=1pt] {\small $58$} (b);

  \end{tikzpicture}

  \end{minipage}
$\;+$\begin{tabular}{c}
  23 \\
  \hline
  $b$ \\
  $d$ \\
  $e$ \\
  $a$ \\
  $c$
\end{tabular}$=\;\;$  \begin{minipage}{1.5in}
  \begin{center}
  $\mathcal{M}_4$
  \end{center}
  \begin{tikzpicture}

  \node[circle,draw,minimum width=0.25in] at (1.5,3) (a) {$a$};
  \node[circle,draw,fill=gray!75, minimum width=0.25in] at (3,1.5) (b) {$b$};
  \node[circle,draw,minimum width=0.25in] at (2.3,-0.5) (c) {$c$};
  \node[circle,draw,fill=gray!75, minimum width=0.25in] at (0.7,-0.5) (d) {$d$};
  \node[circle,draw,minimum width=0.25in] at (0,1.5) (e) {$e$};

  \path[->,draw,thick] (a) to node[fill=white,inner sep=1pt,pos=0.55] {\small $125$} (c);
  \path[->,draw,thick] (a) to node[fill=white,inner sep=1pt,pos=0.55] {\small $11$} (d);
  \path[->,draw,thick] (a) to node[fill=white,inner sep=1pt] {\small $43$} (e);
  \path[->,draw,thick] (b) to node[fill=white,inner sep=1pt] {\small $85$} (a);
  \path[->,draw,thick] (b) to node[fill=white,inner sep=1pt] {\small $93$} (c);
  \path[->,draw,thick] (b) to node[fill=white,inner sep=1pt,pos=0.55] {\small $5$} (d);
  \path[->,draw,thick] (c) to node[fill=white,inner sep=1pt,pos=0.45] {\small $31$} (e);
  \path[->,draw,thick] (d) to node[fill=white,inner sep=1pt] {\small $1$} (c);
  \path[->,draw,thick] (d) to node[fill=white,inner sep=1pt] {\small $37$} (e);
  \path[->,draw,thick] (e) to node[fill=white,inner sep=1pt] {\small $35$} (b);

  \end{tikzpicture}

  \end{minipage}
$=$\begin{tabular}{c}
  7 \\
  \hline
  $d$ \\
  $b$ \\
  $a$ \\
  $c$ \\
  $e$
\end{tabular}$+\;\;$  \begin{minipage}{1.5in}
  \begin{center}
  $\mathcal{M}_5$
  \end{center}
  \begin{tikzpicture}

  \node[circle,draw,minimum width=0.25in] at (1.5,3) (a) {$a$};
  \node[circle,draw,minimum width=0.25in] at (3,1.5) (b) {$b$};
  \node[circle,draw,minimum width=0.25in] at (2.3,-0.5) (c) {$c$};
  \node[circle,draw,fill=gray!75, minimum width=0.25in] at (0.7,-0.5) (d) {$d$};
  \node[circle,draw,minimum width=0.25in] at (0,1.5) (e) {$e$};

  \path[->,draw,thick] (a) to node[fill=white,inner sep=1pt,pos=0.55] {\small $118$} (c);
  \path[->,draw,thick] (a) to node[fill=white,inner sep=1pt,pos=0.55] {\small $18$} (d);
  \path[->,draw,thick] (a) to node[fill=white,inner sep=1pt] {\small $36$} (e);
  \path[->,draw,thick] (b) to node[fill=white,inner sep=1pt] {\small $78$} (a);
  \path[->,draw,thick] (b) to node[fill=white,inner sep=1pt] {\small $86$} (c);
  \path[->,draw,thick] (b) to node[fill=white,inner sep=1pt,pos=0.55] {\small $12$} (d);
  \path[->,draw,thick] (c) to node[fill=white,inner sep=1pt] {\small $6$} (d);
  \path[->,draw,thick] (c) to node[fill=white,inner sep=1pt,pos=0.45] {\small $24$} (e);
  \path[->,draw,thick] (d) to node[fill=white,inner sep=1pt] {\small $30$} (e);
  \path[->,draw,thick] (e) to node[fill=white,inner sep=1pt] {\small $42$} (b);

  \end{tikzpicture}

  \end{minipage}

  \end{center}
\caption{margin graphs used in the proof of Theorem~\ref{MainThm} with defensible candidates shaded gray.  An edge from candidate $x$ to candidate $y$ labeled by $k$ indicates that the margin of $x$ vs.~$y$ is $k$.}\label{KeyFig}
\end{figure}

Now in $\mathbf{P}_1$, $a$ is the only defensible candidate who is not a Condorcet loser (see $\mathcal{M}_1$ in Figure~\ref{KeyFig}). Thus, by the Condorcet loser criterion and Lemma \ref{RefineLem}, $F(\mathbf{P}_1)=\{a\}$. Since $\mathbf{P}_2$ is obtained from $\mathbf{P}_1$ by adding voters all of whom rank $a$ uniquely first,  together $F(\mathbf{P}_1)=\{a\}$ and repeated application of positive involvement imply $a\in F(\mathbf{P}_2)$. Then by resolvability, there is some ``empty''\footnote{Here $\mathbf{S}_2$ is not literally an empty profile; if we already have $F(\mathbf{P}_2)=\{a\}$, then let `$\mathbf{P}_2+\mathbf{S}_2$' denote $\mathbf{P}_2$, `$\mathbf{P}_3+\mathbf{S}_2$' denote $\mathbf{P}_3$, etc. The analogous convention applies to $\mathbf{S}_4$ below if $F(\mathbf{P}_4+\mathbf{S}_2)=\{b\}$.} or single-voter profile $\mathbf{S}_2$ such that ${F(\mathbf{P}_2+\mathbf{S}_2)=\{a\}}$. Since $\mathbf{P}_2+\mathbf{S}_2$ is obtained from $\mathbf{P}_3+\mathbf{S}_2$ by adding voters all of whom rank $d$ uniquely first, together ${F(\mathbf{P}_2+\mathbf{S}_2)=\{a\}}$ and positive involvement imply $d\not\in F(\mathbf{P}_3+\mathbf{S}_2)$. Since any difference between weights of two distinct edges in $\mathcal{M}_3$ is of size at least $4$, no matter which single-voter (or empty) profile $\mathbf{S}_2$ is, $b$ is the only defensible candidate besides $d$ in $\mathbf{P}_3+\mathbf{S}_2$, and $\mathbf{P}_3+\mathbf{S}_2$ satisfies the margin separation assumption of Lemma~\ref{RefineLem},  so $d\not\in F(\mathbf{P}_3+\mathbf{S}_2)$ implies $F(\mathbf{P}_3+\mathbf{S}_2)=\{b\}$ by Lemma~\ref{RefineLem}.  Then since $\mathbf{P}_4+\mathbf{S}_2$ is obtained from $\mathbf{P}_3+\mathbf{S}_2$ by adding voters all of whom rank $b$ uniquely first, together $F(\mathbf{P}_3+\mathbf{S}_2)=\{b\}$ and positive involvement imply $b\in F(\mathbf{P}_4+\mathbf{S}_2)$. Then by resolvability, there is some ``empty'' or single-voter profile $\mathbf{S}_4$ such that $F(\mathbf{P}_4+\mathbf{S}_2+\mathbf{S}_4)=\{b\}$. Since $\mathbf{P}_4+\mathbf{S}_2+\mathbf{S}_4$ is obtained from $\mathbf{P}_5+\mathbf{S}_2+\mathbf{S}_4$ by adding voters all of whom rank $d$ uniquely first,  together $F(\mathbf{P}_4+\mathbf{S}_2+\mathbf{S}_4)=\{b\}$ and positive involvement imply $d\not\in F(\mathbf{P}_5+\mathbf{S}_2+\mathbf{S}_4)$. Finally, since any difference between weights of two distinct edges in $\mathcal{M}_5$ is of size at least $6$, no matter which single-voter (or empty) profiles $\mathbf{S}_2$ and $\mathbf{S}_4$ are, $d$ is the only defensible candidate in $\mathbf{P}_5+\mathbf{S}_2+\mathbf{S}_4$, and $\mathbf{P}_5+\mathbf{S}_2+\mathbf{S}_4$ satisfies the margin separation assumption of Lemma~\ref{RefineLem}, so $F(\mathbf{P}_5+\mathbf{S}_2+\mathbf{S}_4)=\{d\}$ by Lemma~\ref{RefineLem}. Therefore, we have a contradiction.\end{proof}

\begin{remark}Inspection of the proof of Theorem~\ref{MainThm} shows that resolvability can be replaced by \textit{quasi-resoluteness} (in the terminology of \citealt{HP2023}): $|F(\mathbf{P})|=1$ for any profile $\mathbf{P}$ that is \textit{uniquely weighted} in the sense that for all  $x,y,x',y'\in X(\mathbf{P})$ with $x\neq y$ and $x'\neq y'$, if $(x,y)\neq (x',y')$, then $\text{Margin}_\mathbf{P}(x,y)\neq \text{Margin}_\mathbf{P}(x',y')$. It also shows that resolvability can be replaced by \textit{singleton positive involvement} (which Richelson \citeyearpar{Richelson1978} calls \textit{voter adaptability}), according to which $F(\mathbf{P})=\{x\}$ implies $F(\mathbf{P}+\mathbf{P}^x)=\{x\}$ (the analogous point applies to Theorem~\ref{MainThmNI} and what Brandt et al.~\citeyearpar{Brandt2025} call \textit{singleton negative involvement}).\footnote{A slight modification of the proof of Theorem~\ref{MainThm} shows that positive involvement can be replaced by the weaker axiom that if $F(\mathbf{P})=\{x\}$ and $\mathbf{P}^x$ is a profile of new voters all of whom rank $x$ uniquely first, then $x\in F(\mathbf{P}+\mathbf{P}^x)$. First, a variant of Lemma~\ref{RefineLem} holds with the weaker axiom in place of positive involvement, provided we assume $F$ satisfies resolvability and that $\text{Margin}_\mathbf{P}(x,x')>\text{Margin}_\mathbf{P}(y,y')$ implies $\text{Margin}_\mathbf{P}(x,x')>\text{Margin}_\mathbf{P}(y,y')+3$. Now double each voter in each profile in the proof of Theorem~\ref{MainThm}, so all margins double. With the weaker axiom, we derive $F(\mathbf{P}_3+\mathbf{S}_2)\neq \{d\}$ and then $b\in F(\mathbf{P}_3+\mathbf{S}_2)$ by the variant of Lemma~\ref{RefineLem}. Then by resolvability, we add an ``empty'' or single-voter $\mathbf{S}_3$ such that $F(\mathbf{P}_3+\mathbf{S}_2+\mathbf{S}_3)=\{b\}$. From here the proof proceeds as before but to the conclusion that $F(\mathbf{P}_5+\mathbf{S}_2+\mathbf{S}_3+\mathbf{S}_4)\neq \{d\}$. By the variant of Lemma~\ref{RefineLem}, which is still applicable due to the doubling of margins, we have  $F(\mathbf{P}_5+\mathbf{S}_2+\mathbf{S}_3+\mathbf{S}_4)=\{d\}$. Similarly, the negative involvement axiom in Theorem~\ref{MainThmNI} can be replaced by the weaker axiom that if $x\not\in F(\mathbf{P})$ and $\mathbf{P}^{\downarrow x}$ is a profile of new voters all of whom rank $x$ uniquely last, then $F(\mathbf{P}+\mathbf{P}^{\downarrow x})\neq\{x\}$.} This formulation with positive involvement and singleton positive involvement has an implication for probabilistic voting methods if we regard $F(\mathbf{P})$ as the support of a probability distribution: the Condorcet winner and loser criteria are jointly inconsistent with the axiom that adding a voter who ranks $x$ uniquely first never lowers $x$'s probability of winning (from positive to $0$ or from $1$ to less than~$1$).\end{remark}

The same proof strategy can be used to show that for each positive integer $n$, resolvability can be replaced by \textit{$n$-voter resolvability}: for every profile $\mathbf{P}$, if $|F(\mathbf{P})|>1$, then for each $x\in F(\mathbf{P})$, there is a profile $\mathbf{Q}$ with $|V(\mathbf{Q})|\leq n$, $V(\mathbf{P})\cap V(\mathbf{Q})=\varnothing$, and $X(\mathbf{P})=X(\mathbf{Q})$ such that $F(\mathbf{P}+\mathbf{Q})=\{x\}$.\footnote{We could instead consider \textit{$r$-resolvability}: for some $N$ and every $\mathbf{P}$ with $|V(\mathbf{P})|\geq N$, if $|F(\mathbf{P})|>1$, then for each $x\in F(\mathbf{P})$, there is a profile $\mathbf{Q}$ of new voters with $|V(\mathbf{Q})|\leq r \cdot |V(\mathbf{P})|$ such that $F(\mathbf{P}+\mathbf{Q})=\{x\}$. A natural question is: how small can one make $r$ such that $r$-resolvability is consistent with the other axioms? By the proof of Theorem~\ref{MainThm} and a scaling argument as in the proof of Theorem~\ref{nvoter}, $r>1/262$, as $|V(\mathbf{P}_4+\mathbf{S}_2)|\leq 262$ (though this can be improved, e.g., to $1/231$ by Footnote \ref{P1note}).}

\begin{theorem}\label{nvoter} For each positive integer $n$, there is no voting method satisfying the Condorcet winner and loser criteria, positive involvement, and $n$-voter resolvability.
\end{theorem}
\begin{proof} Replace each voter in  $\mathbf{P}_1,\dots,\mathbf{P}_5$ by $n$ copies of that voter, which results in a multiplication of each margin in $\mathcal{M}_1,\dots,\mathcal{M}_5$ by $n$, and then the proof from Theorem \ref{MainThm} goes through as before but now with $|V(\mathbf{S}_2)|\leq n$ and $|V(\mathbf{S}_4)|\leq n$.\end{proof}

These proofs can also be adapted to use \textit{negative involvement} instead of positive involvement: for every profile $\mathbf{P}$, $x\in X(\mathbf{P})$, and profile $\mathbf{P}^{\mathord{\downarrow}x}$ with $V(\mathbf{P})\cap V(\mathbf{P}^{\mathord{\downarrow}x})=\varnothing$ and $X(\mathbf{P})=X(\mathbf{P}^{\mathord{\downarrow}x})$, if $x\not\in F(\mathbf{P})$ and $\mathbf{P}^{\mathord{\downarrow}x}$ consists of a single voter who ranks $x$ uniquely last, then $x\not\in F(\mathbf{P}+\mathbf{P}^{\mathord{\downarrow}x})$.\footnote{There was no need for a separate proof for negative involvement in \citealt{Holliday2024}, since relative to the additional assumption of ordinal margin invariance in that paper, negative involvement is equivalent to positive involvement. However, without ordinal margin invariance, they are not equivalent. For example, Young's \citeyearpar{Young1977} voting method satisfies negative involvement (as well as the Condorcet winner criterion and resolvability) but not positive involvement (see \citealt{Perez2001}, \citealt{Kasper2019}).} Contrapositively, if $x\in F(\mathbf{P}+\mathbf{P}^{\mathord{\downarrow}x})$, then $x\in F(\mathbf{P})$.

\begin{theorem}\label{MainThmNI} For each positive integer $n$, there is no voting method satisfying the Condorcet winner and loser criteria, negative involvement, and $n$-voter resolvability.\end{theorem}

\begin{proof} We give the proof for $n=1$. A scaling argument exactly analogous to that in the proof of Theorem~\ref{nvoter} handles $n>1$.  First, observe that a variant of Lemma \ref{RefineLem} holds for negative involvement (cf.~\citealt[Lemma 1]{Perez2001}): Let $F$ be a voting method satisfying negative involvement and the Condorcet winner criterion. Let $\mathbf{P}$  be a profile in which for all $x,x',y,y'\in X(\mathbf{P})$ with $x\neq x'$ and $y\neq y'$, if $\text{Margin}_\mathbf{P}(x,x')>\text{Margin}_\mathbf{P}(y,y')$, then $\text{Margin}_\mathbf{P}(x,x')>\text{Margin}_\mathbf{P}(y,y')+1$. Additionally assume that for each $x\in X(\mathbf{P})\setminus D(\mathbf{P})$, there is a $y\in X(\mathbf{P})$ such that for every $z\in X(\mathbf{P})$, $\text{Margin}_\mathbf{P}(z,y)< \text{Margin}_\mathbf{P}(y,x)$, and where $k= \max\{\text{Margin}_\mathbf{P}(z,y) \mid z\in X(\mathbf{P}),z\neq y\}$, there are at least $k+1$ voters whose rankings have $y$ uniquely second-to-last, followed by $x$ uniquely last. Then $F(\mathbf{P})\subseteq D(\mathbf{P})$. The proof is the same as that of Lemma~\ref{RefineLem} except we obtain $\mathbf{P}'$ from $\mathbf{P}$ by \textit{removing} $k+1$ voters of the kind described above. Then since the proof of Lemma~\ref{RefineLem} begins with $x\in F(\mathbf{P})$, repeated application of negative involvement in its contrapositive form implies $x\in F(\mathbf{P}')$. Removing a ranking $L$ has the same effect on margins as adding the reversed ranking $L^{-1}$, so as in the proof of Lemma~\ref{RefineLem}, $y$ is the Condorcet winner in $\mathbf{P}'$, yielding a contradiction.

Now we adapt the proof of Theorem \ref{MainThm} by adding 147 copies of each linear order of $X(\mathbf{P}_1)$ to $\mathbf{P}_1$, which ensures that the variant of Lemma~\ref{RefineLem} is applicable when needed in the proof.\footnote{By starting with at least $147$ copies of each linear order in $\mathbf{P}_1$, we ensure that there are at least $121$ copies of each linear order in $\mathbf{P}_3+\mathbf{S}_2$ and in $\mathbf{P}_5+\mathbf{S}_2+\mathbf{S}_4$, which is greater than the maximal margins in those profiles, so one immediately sees that the variant of Lemma~\ref{RefineLem} is applicable to those profiles, without any reasoning about the numbers of rankings of specific types (reasoning that could be used to replace $147$ with a smaller number). Of course, one can use an integer linear program to synthesize a new starting profile with far fewer voters (indeed, $477$ voters) such that all constraints needed for the proof~still~hold.} Then instead of adding 26 voters with the ranking $adbec$ to obtain $\mathbf{P}_2$, we  \textit{remove} 26 voters with $cebda$; instead of removing 7 voters with  $daceb$ from $\mathbf{P}_2$ to obtain $\mathbf{P}_3$, we \textit{add} 7 voters with $becad$; and so on for $\mathbf{P}_4$ (removing 23 voters with $caedb$) and $\mathbf{P}_5$ (adding 7 voters with $ecabd$). The margins of $\mathcal{M}_1,\dots,\mathcal{M}_5$ remain as in the proof of Theorem~\ref{MainThm}, and then the rest of the proof is the same except that we (i) invoke the variant of Lemma~\ref{RefineLem} for steps involving defensibility and (ii) use negative involvement instead of positive involvement.\end{proof}

Finally, in response to a previous version of this note with the results above, Dominik Peters (personal communication) suggested that analogous impossibilities might hold with \textit{independence of clones} (\citealt{Tideman1987}) in place of the Condorcet loser criterion. Indeed, we will confirm below that such a swap in Theorems~\ref{MainThm}~and~\ref{nvoter} is possible, as well as in Theorem \ref{MainThmNI} if we add the axiom of \textit{block preservation} (\citealt{HP2025}): if the profile $\mathbf{L}$ consists of exactly one copy of each linear order of $X(\mathbf{P})$, then $F(\mathbf{P})\subseteq F(\mathbf{P}+\mathbf{L})$. 

A set $C\subsetneq X(\mathbf{P})$  with $|C|\geq 2$ is a set of \textit{clones} in $\mathbf{P}$ if for each $i\in V(\mathbf{P})$, $x\in X(\mathbf{P})\setminus C$, and $c,c'\in C$, if $i$ ranks $x$ above (resp.~below) $c$, then $i$ ranks $x$ above (resp.~below) $c'$. Given a profile $\mathbf{P}$ and $c\in X(\mathbf{P})$, let $\mathbf{P}_{-c}$ be the profile with $X(\mathbf{P}_{-c})=X(\mathbf{P})\setminus \{c\}$ and $V(\mathbf{P}_{-c})=V(\mathbf{P})$ where for each $i\in V(\mathbf{P})$, $\mathbf{P}_{-c}(i)$ is the restriction of $\mathbf{P}(i)$ to $X(\mathbf{P}_{-c})$. Then $F$ satisfies independence of clones if for any profile $\mathbf{P}$, set $C$ of clones, $c\in C$, and $x\in X(\mathbf{P})\setminus C$, (i) $x\in F(\mathbf{P})$ iff $x\in F(\mathbf{P}_{-c})$, and (ii) $C\cap F(\mathbf{P})\neq\varnothing$ iff $C\cap F(\mathbf{P}_{-c})\neq\varnothing$.

\begin{theorem}\label{CloneThm}$\,$
\begin{enumerate}
\item\label{CloneThm1} For each positive integer $n$, there is no voting method satisfying the Condorcet winner criterion, independence of clones, positive involvement, and $n$-voter resolvability. 
\item\label{CloneThm2} The same holds with positive involvement replaced by negative involvement plus block preservation.
\end{enumerate}
\end{theorem}

\begin{proof} We prove the $n=1$ case. Reasoning as in the proof of Theorem~\ref{nvoter} handles $n>1$. 

For part~\ref{CloneThm1}, consider a profile $\mathbf{Q}_1$ as in Figure~\ref{Q1}, wherein $\{a,b,c,e\}$ is a set of clones. Its margin graph is shown as $\mathcal{M}_1$ in Figure~\ref{KeyFig2}. By Lemma~\ref{RefineLem}, $F(\mathbf{Q}_1)\subseteq \{a,d\}$. By the Condorcet winner criterion, $F(\mathbf{Q}_{1-b})=\{a\}$, so by independence of clones (part (i) or (ii)) and $F(\mathbf{Q}_1)\subseteq \{a,d\}$, we have $a\in F(\mathbf{Q}_1)$. Then obtain $\mathbf{Q}_2,\dots,\mathbf{Q}_5$ from $\mathbf{Q}_1$ by adding and subtracting the profiles in Figure~\ref{KeyFig2}. The rest of the proof is just like that of Theorem~\ref{MainThm} but with $\mathbf{Q}_i$ instead of~$\mathbf{P}_i$. 

For part~\ref{CloneThm2}, by block preservation we can add as many blocks of all linear orders as we wish to $\mathbf{Q}_1$ to obtain a $\mathbf{Q}_1^+$ with $F(\mathbf{Q}_1)\subseteq F(\mathbf{Q}_1^+)\subseteq \{a,d\}$ by the variant of Lemma~\ref{RefineLem} in the proof of Theorem~\ref{MainThmNI}. Hence $F(\mathbf{Q}_1)\subseteq \{a,d\}$, which implies $a\in F(\mathbf{Q}_1)$ as above for part~\ref{CloneThm1} and therefore $a\in F(\mathbf{Q}_1^+)$. Then starting from $\mathbf{Q}_1^+$, we reason with negative involvement instead of positive involvement as in the proof of Theorem~\ref{MainThmNI}.\end{proof}

\noindent Thus, the tradeoff noted in Section~\ref{Intro} is shown by several inconsistent sets of axioms, summarized in Table~\ref{SummaryTable}.

\begin{figure}[h]
\begin{center}
  \begin{tabular}{cccccc}
  62 & 60 & 42 & 23 & 19 & 3\\
  \hline
  $d$ & $c$ & $d$ & $a$ & $e$ & $c$ \\
  $b$ & $b$ & $e$ & $e$ & $c$ & $e$ \\
  $a$ & $a$ & $a$ & $c$ & $b$ & $b$ \\
  $e$ & $e$ & $c$ & $b$ & $a$ & $a$ \\
  $c$ & $d$ & $b$ & $d$ & $d$ & $d$ \\
  \end{tabular}
  \end{center}
\caption{the profile $\mathbf{Q}_1$ used in the proof of Theorem \ref{CloneThm}.}\label{Q1}
\end{figure}

\begin{figure}[h!]
 \begin{center}

  \begin{minipage}{1.5in}
  \begin{center}
  $\mathcal{M}_1$
  \end{center}
  \begin{tikzpicture}

  \node[circle,draw,fill=gray!75, minimum width=0.25in] at (1.5,3) (a) {$a$};
  \node[circle,draw,minimum width=0.25in] at (3,1.5) (b) {$b$};
  \node[circle,draw,minimum width=0.25in] at (2.3,-0.5) (c) {$c$};
  \node[circle,draw,fill=gray!75, minimum width=0.25in] at (0.7,-0.5) (d) {$d$};
  \node[circle,draw,minimum width=0.25in] at (0,1.5) (e) {$e$};

  \path[->,draw,thick] (a) to node[fill=white,inner sep=1pt,pos=0.55] {\small $45$} (c);
  \path[->,draw,thick] (a) to node[fill=white,inner sep=1pt,pos=0.55] {\small $1$} (d);
  \path[->,draw,thick] (a) to node[fill=white,inner sep=1pt] {\small $81$} (e);
  \path[->,draw,thick] (b) to node[fill=white,inner sep=1pt] {\small $79$} (a);
  \path[->,draw,thick] (b) to node[fill=white,inner sep=1pt,pos=0.55] {\small $1$} (d);
  \path[->,draw,thick] (b) to node[fill=white,inner sep=1pt] {\small $35$} (e);
  \path[->,draw,thick] (c) to node[fill=white,inner sep=1pt] {\small $85$} (b);
  \path[->,draw,thick] (c) to node[fill=white,inner sep=1pt] {\small $1$} (d);
  \path[->,draw,thick] (e) to node[fill=white,inner sep=1pt,pos=0.55] {\small $83$} (c);
  \path[->,draw,thick] (e) to node[fill=white,inner sep=1pt] {\small $1$} (d);

  \end{tikzpicture}

  \end{minipage}
$\;+$\begin{tabular}{ccc}
  16 & 7 & 2 \\
  \hline
  $a$ & $a$ & $a$ \\
  $d$ & $e$ & $e$ \\
  $b$ & $d$ & $b$ \\
  $c$ & $b$ & $d$ \\
  $e$ & $c$ & $c$
\end{tabular}$=\;\;$  \begin{minipage}{1.5in}
  \begin{center}
  $\mathcal{M}_2$
  \end{center}
  \begin{tikzpicture}

  \node[circle,draw,fill=gray!75, minimum width=0.25in] at (1.5,3) (a) {$a$};
  \node[circle,draw,fill=gray!75, minimum width=0.25in] at (3,1.5) (b) {$b$};
  \node[circle,draw,minimum width=0.25in] at (2.3,-0.5) (c) {$c$};
  \node[circle,draw,fill=gray!75, minimum width=0.25in] at (0.7,-0.5) (d) {$d$};
  \node[circle,draw,minimum width=0.25in] at (0,1.5) (e) {$e$};

  \path[->,draw,thick] (a) to node[fill=white,inner sep=1pt,pos=0.55] {\small $70$} (c);
  \path[->,draw,thick] (a) to node[fill=white,inner sep=1pt,pos=0.55] {\small $26$} (d);
  \path[->,draw,thick] (a) to node[fill=white,inner sep=1pt] {\small $106$} (e);
  \path[->,draw,thick] (b) to node[fill=white,inner sep=1pt] {\small $54$} (a);
  \path[->,draw,thick] (b) to node[fill=white,inner sep=1pt] {\small $42$} (e);
  \path[->,draw,thick] (c) to node[fill=white,inner sep=1pt] {\small $60$} (b);
  \path[->,draw,thick] (d) to node[fill=white,inner sep=1pt,pos=0.45] {\small $20$} (b);
  \path[->,draw,thick] (d) to node[fill=white,inner sep=1pt] {\small $24$} (c);
  \path[->,draw,thick] (d) to node[fill=white,inner sep=1pt] {\small $6$} (e);
  \path[->,draw,thick] (e) to node[fill=white,inner sep=1pt,pos=0.55] {\small $76$} (c);

  \end{tikzpicture}

  \end{minipage}
$=$\begin{tabular}{c}
  5 \\
  \hline
  $d$ \\
  $e$ \\
  $a$ \\
  $c$ \\
  $b$
\end{tabular}$+\;\;$  \begin{minipage}{1.5in}
  \begin{center}
  $\mathcal{M}_3$
  \end{center}
  \begin{tikzpicture}

  \node[circle,draw,minimum width=0.25in] at (1.5,3) (a) {$a$};
  \node[circle,draw,fill=gray!75, minimum width=0.25in] at (3,1.5) (b) {$b$};
  \node[circle,draw,minimum width=0.25in] at (2.3,-0.5) (c) {$c$};
  \node[circle,draw,fill=gray!75, minimum width=0.25in] at (0.7,-0.5) (d) {$d$};
  \node[circle,draw,minimum width=0.25in] at (0,1.5) (e) {$e$};

  \path[->,draw,thick] (a) to node[fill=white,inner sep=1pt,pos=0.55] {\small $65$} (c);
  \path[->,draw,thick] (a) to node[fill=white,inner sep=1pt,pos=0.55] {\small $31$} (d);
  \path[->,draw,thick] (a) to node[fill=white,inner sep=1pt] {\small $111$} (e);
  \path[->,draw,thick] (b) to node[fill=white,inner sep=1pt] {\small $59$} (a);
  \path[->,draw,thick] (b) to node[fill=white,inner sep=1pt] {\small $47$} (e);
  \path[->,draw,thick] (c) to node[fill=white,inner sep=1pt] {\small $55$} (b);
  \path[->,draw,thick] (d) to node[fill=white,inner sep=1pt,pos=0.45] {\small $15$} (b);
  \path[->,draw,thick] (d) to node[fill=white,inner sep=1pt] {\small $19$} (c);
  \path[->,draw,thick] (d) to node[fill=white,inner sep=1pt] {\small $1$} (e);
  \path[->,draw,thick] (e) to node[fill=white,inner sep=1pt,pos=0.55] {\small $71$} (c);

  \end{tikzpicture}

  \end{minipage}

  \vspace{0.5cm}

  \begin{minipage}{1.5in}
  \begin{center}
  $\mathcal{M}_3$
  \end{center}
  \begin{tikzpicture}

  \node[circle,draw,minimum width=0.25in] at (1.5,3) (a) {$a$};
  \node[circle,draw,fill=gray!75, minimum width=0.25in] at (3,1.5) (b) {$b$};
  \node[circle,draw,minimum width=0.25in] at (2.3,-0.5) (c) {$c$};
  \node[circle,draw,fill=gray!75, minimum width=0.25in] at (0.7,-0.5) (d) {$d$};
  \node[circle,draw,minimum width=0.25in] at (0,1.5) (e) {$e$};

  \path[->,draw,thick] (a) to node[fill=white,inner sep=1pt,pos=0.55] {\small $65$} (c);
  \path[->,draw,thick] (a) to node[fill=white,inner sep=1pt,pos=0.55] {\small $31$} (d);
  \path[->,draw,thick] (a) to node[fill=white,inner sep=1pt] {\small $111$} (e);
  \path[->,draw,thick] (b) to node[fill=white,inner sep=1pt] {\small $59$} (a);
  \path[->,draw,thick] (b) to node[fill=white,inner sep=1pt] {\small $47$} (e);
  \path[->,draw,thick] (c) to node[fill=white,inner sep=1pt] {\small $55$} (b);
  \path[->,draw,thick] (d) to node[fill=white,inner sep=1pt,pos=0.45] {\small $15$} (b);
  \path[->,draw,thick] (d) to node[fill=white,inner sep=1pt] {\small $19$} (c);
  \path[->,draw,thick] (d) to node[fill=white,inner sep=1pt] {\small $1$} (e);
  \path[->,draw,thick] (e) to node[fill=white,inner sep=1pt,pos=0.55] {\small $71$} (c);

  \end{tikzpicture}

  \end{minipage}$\,\;\;\;\,+\;\;\;$\begin{tabular}{c}
  20 \\
  \hline
  $b$ \\
  $d$ \\
  $c$ \\
  $a$ \\
  $e$
\end{tabular}$\;\;\;=\;\;\;\;\;\,\,$  \begin{minipage}{1.5in}
  \begin{center}
  $\mathcal{M}_4$
  \end{center}
  \begin{tikzpicture}

  \node[circle,draw,minimum width=0.25in] at (1.5,3) (a) {$a$};
  \node[circle,draw,fill=gray!75, minimum width=0.25in] at (3,1.5) (b) {$b$};
  \node[circle,draw,minimum width=0.25in] at (2.3,-0.5) (c) {$c$};
  \node[circle,draw,fill=gray!75, minimum width=0.25in] at (0.7,-0.5) (d) {$d$};
  \node[circle,draw,minimum width=0.25in] at (0,1.5) (e) {$e$};

  \path[->,draw,thick] (a) to node[fill=white,inner sep=1pt,pos=0.55] {\small $45$} (c);
  \path[->,draw,thick] (a) to node[fill=white,inner sep=1pt,pos=0.55] {\small $11$} (d);
  \path[->,draw,thick] (a) to node[fill=white,inner sep=1pt] {\small $131$} (e);
  \path[->,draw,thick] (b) to node[fill=white,inner sep=1pt] {\small $79$} (a);
  \path[->,draw,thick] (b) to node[fill=white,inner sep=1pt,pos=0.55] {\small $5$} (d);
  \path[->,draw,thick] (b) to node[fill=white,inner sep=1pt] {\small $67$} (e);
  \path[->,draw,thick] (c) to node[fill=white,inner sep=1pt] {\small $35$} (b);
  \path[->,draw,thick] (d) to node[fill=white,inner sep=1pt] {\small $39$} (c);
  \path[->,draw,thick] (d) to node[fill=white,inner sep=1pt] {\small $21$} (e);
  \path[->,draw,thick] (e) to node[fill=white,inner sep=1pt,pos=0.55] {\small $51$} (c);

  \end{tikzpicture}

  \end{minipage}
$=$\begin{tabular}{c}
  11 \\
  \hline
  $d$ \\
  $b$ \\
  $a$ \\
  $e$ \\
  $c$
\end{tabular}$+\;\;$  \begin{minipage}{1.5in}
  \begin{center}
  $\mathcal{M}_5$
  \end{center}
  \begin{tikzpicture}

  \node[circle,draw,minimum width=0.25in] at (1.5,3) (a) {$a$};
  \node[circle,draw,minimum width=0.25in] at (3,1.5) (b) {$b$};
  \node[circle,draw,minimum width=0.25in] at (2.3,-0.5) (c) {$c$};
  \node[circle,draw,fill=gray!75, minimum width=0.25in] at (0.7,-0.5) (d) {$d$};
  \node[circle,draw,minimum width=0.25in] at (0,1.5) (e) {$e$};

  \path[->,draw,thick] (a) to node[fill=white,inner sep=1pt,pos=0.55] {\small $34$} (c);
  \path[->,draw,thick] (a) to node[fill=white,inner sep=1pt,pos=0.55] {\small $22$} (d);
  \path[->,draw,thick] (a) to node[fill=white,inner sep=1pt] {\small $120$} (e);
  \path[->,draw,thick] (b) to node[fill=white,inner sep=1pt] {\small $68$} (a);
  \path[->,draw,thick] (b) to node[fill=white,inner sep=1pt,pos=0.55] {\small $16$} (d);
  \path[->,draw,thick] (b) to node[fill=white,inner sep=1pt] {\small $56$} (e);
  \path[->,draw,thick] (c) to node[fill=white,inner sep=1pt] {\small $46$} (b);
  \path[->,draw,thick] (d) to node[fill=white,inner sep=1pt] {\small $28$} (c);
  \path[->,draw,thick] (d) to node[fill=white,inner sep=1pt] {\small $10$} (e);
  \path[->,draw,thick] (e) to node[fill=white,inner sep=1pt,pos=0.55] {\small $40$} (c);

  \end{tikzpicture}

  \end{minipage}

  \end{center}
\caption{margin graphs used in the proof of Theorem~\ref{CloneThm} with defensible candidates shaded gray.}\label{KeyFig2}
\end{figure}
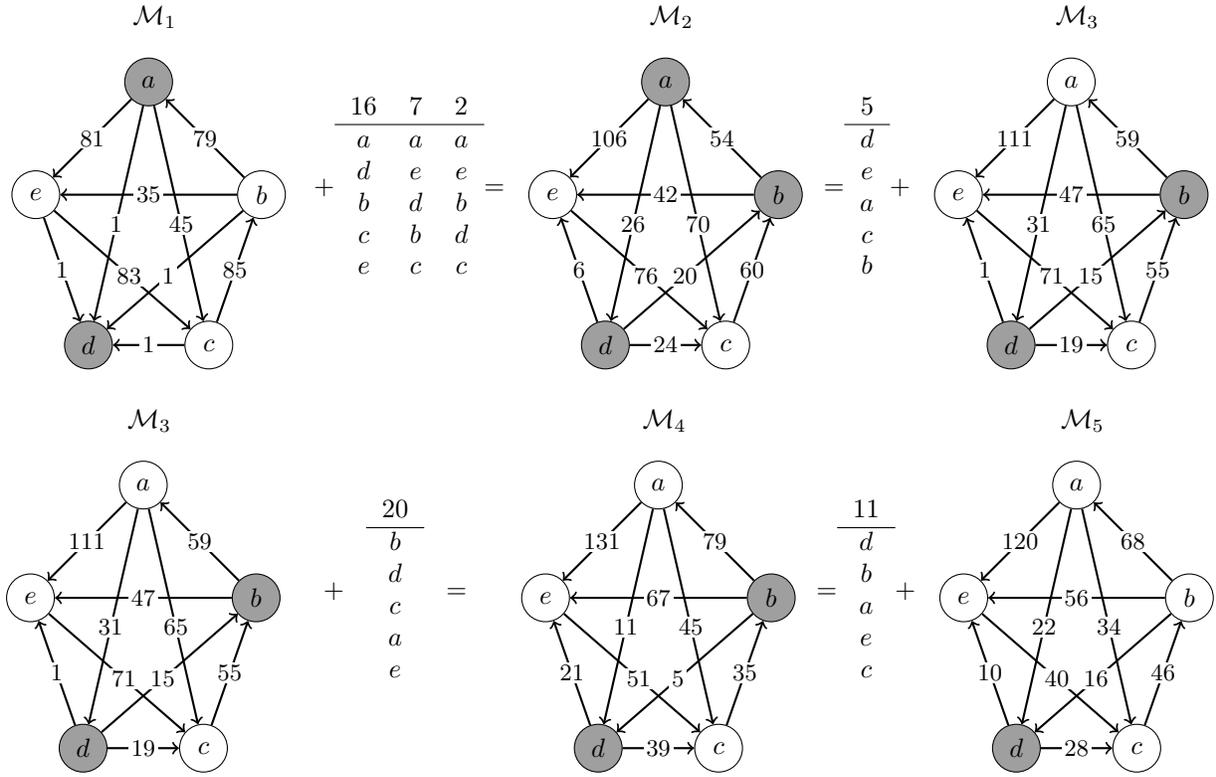

\begin{table}[h!]
\begin{center}
{\small\setlength{\tabcolsep}{5pt}
\begin{tabular}{l|ccccccc}
 & Condorcet & Condorcet & positive & negative & resolvability & independence & block \\
& winner & loser & involvement & involvement & & of clones & preservation  \\
\hline 
Theorem \ref{MainThm} & $\circleddash$ & $\circleddash$ & $\circleddash$ & & $\circleddash$ &  \\
\hline
Theorem \ref{MainThmNI} & $\circleddash$ & $\circleddash$ & & $\circleddash$ & $\circleddash$\\
\hline
Theorem \ref{CloneThm}.\ref{CloneThm1} & $\circleddash$ & & $\circleddash$ & & $\circleddash$ & $\circleddash$ \\
\hline
Theorem \ref{CloneThm}.\ref{CloneThm2} &  $\circleddash$ &  &  & $\circleddash$ & $\circleddash$ & $\circleddash$ & $\circleddash$  \\
\end{tabular}}
\end{center}
\caption{summary of the main impossibility theorems. For each theorem, the axioms that are jointly inconsistent are indicated by circled marks in the row for that theorem. In each case, resolvability can be replaced by $n$-voter resolvability.}\label{SummaryTable}
\end{table}

\newpage

\subsection*{Acknowledgements}

I thank Diego Bejarano and Emma Boniface for helpful discussions of the problem posed at the end of \citealt{Holliday2024}, as well as Dominik Peters and an anonymous referee for helpful feedback on this paper.

\bibliographystyle{plainnat}
\bibliography{social}
  
\end{document}